\documentclass[pre,twocolumn,showpacs]{revtex4-1}
\usepackage{graphicx}
\usepackage{amsmath}
\usepackage{amsthm}
\usepackage{bm}

\newtheorem{theorem}{Theorem}
\newtheorem{lemma}{Lemma}
\newtheorem{corollary}{Corollary}

\begin{document}

\title{Jammed systems of oriented dimers always percolate on hypercubic lattices}
\author{Zbigniew Koza}\email{zbigniew.koza@uwr.edu.pl}
\author{Grzegorz Kondrat}
\affiliation{Faculty of Physics and Astronomy, University of Wroc{\l}aw, 50-204
Wroc{\l}aw,
Poland}

\date{\today}

\begin{abstract}
Random sequential adsorption (RSA) is a standard method of modeling adsorption
of large molecules at the liquid-solid interface.
Here we consider jammed states of the RSA process
of nonoverlapping dimers  (objects occupying two  nearest-neighbor lattice sites) in a hypercubic lattice of arbitrary
space dimension $D \ge 2$. We show that each dimer in such a state
belongs to a percolating cluster.
\end{abstract}


\maketitle


\section{Introduction\label{sec:intro}}

Adsorption of finite size objects plays an important role in various processes in biology,
science and technology. Among them are production of conducting nanocomposites
\cite{Mutiso2015}, water purification \cite{Dabrowski2001} and protein adsorption
on the liquid-solid interface \cite{Talbot2000}, to mention a few.
Often the circumstances allow one to consider these processes as
irreversible, e.g., when a monolayer is being formed
on a target surface~\cite{Feder1980a}.
Introduced by Feder~\cite{Feder1980a}, the random sequential adsorption
(RSA) serves as a paradigm in modeling such processes. In the basic setup
of this method a sequence of identical  extended geometrical objects
(discs, rectangles, line segments, or more complex shapes) is
irreversibly placed one by one at random position and with a random orientation
on an initially empty surface \cite{Feder1980a,Evans1993,Talbot2000}.
A single adsorption trial is successful if the new object does not overlap
with any previously deposited one, otherwise the trial is rejected.
Once an object is adsorbed, it stays at the same position forever.
In this process the rate of successful
attempts drops down steadily as the area accessible for adsorption decreases.
At some moment the dynamics stops, as there is no free place that can accommodate
one more object---the system has reached the jamming limit~\cite{Torquato2010}.

Numerous variants of this simple model have been studied, including deposition of 
various object shapes (discs \cite{Feder1980a}, ellipses \cite{Viot1992},
squares \cite{Brosilow1991,Nakamura1986,Ramirez-Pastor19},
rectangles \cite{Viot1992,Porto2000,Vandewalle2000,Kondrat2001,Tarasevich2012,
Tarasevich2015,Centres2015}, spheres \cite{Meakin1992}, spherocylinders \cite{Schilling2015},
cubes \cite{Buchini19}, square tiles \cite{Pasinetti19}, line segments of zero width
\cite{Provatas2000}), shape flexibility \cite{Adamczyk2008,Kondrat2002},
polydispersity \cite{Lee2000,Nigro2013,Chatterjee2015},
imperfect substrates \cite{Tarasevich2015,Centres2015},
post-adsorption dynamics \cite{Budinski2001,Lebovka2017},
as well as partial \cite{Balberg1987} or full object overlapping \cite{Torquato2012}.

Two  quantities characterize an RSA process: the jamming and
percolation thresholds.
The jamming threshold $0 < c_\mathrm{j} \le 1$ is defined
as the ratio of the area of the occupied space to the area of the whole substrate,
and measures the density of the adsorbate when the RSA process
has come to an end. The percolation threshold $0 < c_\mathrm{p} \le c_\mathrm{j}$
is similar to $c_\mathrm{j}$ except that it is determined
at the moment when the adsorbed objects for the first time form a connected network
that spans the whole system~\cite{Stauffer1994}.


The relation between percolation and jamming thresholds was a subject of many studies.
For the RSA of squares of size $k\times k$  on a square lattice,
percolation is observed only for $k\leq 3$ \cite{Nakamura1986,Ramirez-Pastor19}.
This was generalized to arbitrary rectangles $k_1 \times k_2$,
and the percolation was found to appear only if
$c_\mathrm{j}$ is greater
than the percolation threshold of overlapping disks on a plane \cite{Porto2000}.
Simulations suggest that needles $1 \times 1 \times k$
always percolate on the cubic lattice \cite{Garcia2013,Garcia2015},
with $c_\mathrm{p}/c_\mathrm{j}\rightarrow 0$ as $k\rightarrow\infty$.
There were also studies of RSA with other cuboids:
plates of size $1 \times k \times k$   (all shapes percolate) \cite{Pasinetti19}
and cubes $k \times k \times k$  (shapes
with $k>16$ do not percolate) \cite{Buchini19}. In all these studies the problem of whether
some RSA systems percolate was investigated in  the thermodynamic limit of
the system size going to infinity. In most of the models the probability that
a finite system at a jammed state contains a percolating cluster is neither
0 nor 1, and only in the thermodynamic limit does it converge to 0 (no percolation)
or 1 ($c_\mathrm{p}$ is well-defined).

The case that has attracted the most attention
in the context of the value of
$c_\mathrm{p}/c_\mathrm{j}$ is the RSA of rectangles $1 \times k$ (``needles'')
on a square lattice. Initially it was claimed that this ratio
is independent of $k$~\cite{Vandewalle2000}.
If true, it would imply some deep relation between the
two apparently different phenomena of percolation and jamming.
However, further studies showed that while
the ratio $c_\mathrm{p}/c_\mathrm{j}$ stays almost constant
for short needles, for longer $k$ ($15 \leq k \leq 45$)
it grows as $0.50 + 0.13 \log k$~\cite{Kondrat2001}.
Extrapolation of this phenomenological formula to $k\to\infty$,
supported by numerical data for much longer needles, $k\leq 512$,
led to the hypothesis about the percolation breakdown for sufficiently
long needles \cite{Tarasevich2012}.
Indeed, if $c_\mathrm{p}/c_\mathrm{j}$ increases logarithmically with $k$,
a critical length $k_*$ must exist such that $c_\mathrm{p}/c_\mathrm{j}>1$
for $k>k_*$, which is impossible, as $c_\mathrm{p} \le c_\mathrm{j}$ must by definition
hold for all models where $c_\mathrm{p}$ is well-defined.
It was therefore suggested that if the needles are sufficiently
long ($k\geq k_*$) the system jams before it can percolate
and $c_\mathrm{p}$ is undefined.
Several subsequent studies of RSA
on imperfect substrates supported this observation \cite{Tarasevich2015,Centres2015},
all predicting the value of the critical length $k_*$
to be of order of several thousand lattice constants.
However, this conjecture was disproved  when a rigorous proof was given
that percolation sets in for all needle lengths~\cite{Kondrat17}.  Subsequently,
extensive computer simulations revealed that for very long needles the ratio
$c_\mathrm{p}/c_\mathrm{j}$ actually departs from the logarithmic dependence
on $k$ \cite{Slutskii18} .


Here we address the question of existence of the percolation threshold for RSA
of dimers (``needles'' of length 2) on an arbitrary $D$-dimensional hypercubic lattice.
While some simulation results for this model have already been  obtained
for $D = 2,3$~\cite{Loscar2006,Tarasevich2007,Lebrecht13},
our goal is to tackle the problem rigorously. Our main result states
that percolation in this model is not a statistical property characterizing
the ensemble of RSA realizations in the thermodynamic limit,
but occurs for \emph{every} RSA process in this model, including finite systems.
In other words, one could drop ``random'' in the definition of RSA and still
be sure that each complete realization of the process will contain a percolating cluster.


\section{Theorem and its proof\label{sec:proof}}

We will prove the following

\begin{theorem}\label{theorem::1}
Every jammed configuration of nonoverlapping dimers on a
finite $D$-dimensional hypercubic lattice ($D \ge 2)$ contains a connected cluster spanning
two opposite edges of the lattice.
\end{theorem}

In this theorem, a finite hypercubic lattice is a hypercuboid
$V  \equiv Z_1\times Z_2\times\ldots\times  Z_D$,
where $Z_i = \{1,2,\ldots L_i\}$ for arbitrary integers $L_i \ge 1$, $i=1,2,\ldots,D$,
of which  at least  one is $\ge 2$ so that $V$ can contain at least one dimer.
A dimer is a subset of $V$ made of two adjacent lattice nodes (nearest neighbors).
Two dimers,  $d$ and $d'$ are  connected directly if the shortest  distance between $x\in d$
and $x'\in d'$ is 1. A cluster is a set of the nodes occupied by all dimers
connected directly or indirectly via a sequence of direct connections (Fig.~\ref{fig:def}).
Since the theorem is trivial if any $L_i = 1$, henceforth we will assume
$L_i \ge 2$, $i = 1, \ldots, D$.

\begin{figure}
	\includegraphics[width=0.85\columnwidth]{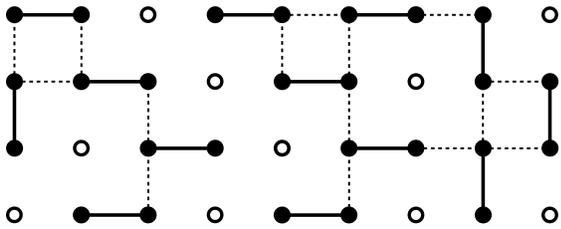}
	\caption{ 	\label{fig:def}
		A jammed configuration of dimers for space dimension $D=2$ and
		system size $L_1 = 9$ and $L_2=4$. Solid lines represent the dimers,
		dashed lines---connections between dimers forming clusters.
		The  dots represent the occupied (solid symbols)
		and unoccupied (open symbols) nodes. }
\end{figure}



We start by temporarily getting rid of finite  boundaries and consider a
jamming coverage of the infinite lattice $Z^D$ by nonoverlapping dimers. That is,
we assume that any lattice node in $Z^D$ is either unoccupied or belongs
to exactly one dimer, and all nearest neighbors of any unoccupied node
belong to some dimer. The latter condition ensures that no dimer
can be added to the system without overlapping  (the system is jammed).
We start from asking a question: can a jammed configuration in an infinite
system contain a finite cluster?

Let us assume that there exists a jamming configuration of
nonoverlapping dimers on $Z^D$ that comprises a finite cluster.
Let this cluster be denoted by $C$ and let $F$ (``full'') denote the union
of $C$ and all unoccupied lattice nodes completely surrounded by the nodes from
$C$ (``holes'').
The set of the edge nodes of $F$ will be denoted $\Gamma$, and the set
of the nodes adjacent to $\Gamma$ and not in $F$ will be denoted as
$\Gamma'$ (by definition, any node in $\Gamma'$ is unoccupied).
The components of each lattice node are integers, so we
can assign to it the parity, even or odd, of their sum.

\begin{lemma}\label{lemma::parity}
  If a finite cluster exists in an infinite system, all nodes of its outer border $\Gamma$ share the same parity.
\end{lemma}
\begin{proof}
Take an arbitrary two-dimensional cross-section of the system
that cuts through the cluster and denote the resulting cross-section of $F$
as $\tilde F$.
This set need not be a single, connected cluster (with nearest-neighbor connections).
Even if it consists of several clusters, their number is finite and
their edge nodes belong to $\Gamma$, whereas all their out-adjacent
nearest neighbors are in $\Gamma'$.
In  general, the sites surrounding a finite cluster formed in 2D by nearest-neighbor
connections form another,  single cluster with nearest- or next-neighbor connections
\cite{Koza2019}.
In the case of $\tilde F$, the nearest-neighbor connections in $\Gamma'$
are excluded by the condition that the system is jammed, so only next-nearest
connections are possible. These, however, preserve the node parity
on a planar cross-section.
Any two nodes in $\Gamma'$ may by connected by jumps to a next-neighbor node
on some two-dimensional cross-sections of the system. Each jump preserves the parity.
Since the nodes in $\Gamma$ are nearest neighbors of some nodes in $\Gamma'$,
all nodes in $\Gamma$ must share the same parity.
\end{proof}

The same reasoning can be applied to infinite clusters or clusters
inside finite systems except that in these cases $\Gamma'$ may
consist of several disconnected subsets, each of arbitrary parity.
Thus, in the general case of a jammed configuration of nonoverlapping
dimers on finite or infinite system, nodes in $\Gamma$ that
are nearest neighbors of the same node in $\Gamma'$ are of the same parity.

We are now ready to prove that the finite clusters to which
Lemma~\ref{lemma::parity} explicitly refers actually do not exist:
\begin{lemma}\label{lemma::infinity}
	Any cluster of nonoverlapping dimers in a jamming coverage of $Z^D$ is infinite.
\end{lemma}

\begin{proof}
Assume that $C$ is finite. Let $\bm{n}$ denote a node, and $n_1,\dots,n_D$ its components.
We also define versors
$\bm\varepsilon_i = \{\delta_{1,i},\delta_{2,i},\ldots,\delta_{D,i},\} $,
$i=1,\ldots,D$, where $\delta_{i,j}$ is the Kronecker delta.
Let $m = \max_{\bm{n}\in C}(n_1+\cdots+n_D)$.
We define two  nonempty sets,
$A = \{\bm{n} \in C: n_1+\cdots+n_D = m\}$ and
$B = \{\bm{n} \in C: n_1+\cdots+n_D = m-1\}$.
By construction, all nodes from $A$ belong to the border of the cluster. However,
since the parity of nodes from $B$ is different from that in $A$,
by Lemma~\ref{lemma::parity} none of the elements of $B$
is a border node.  With these observations we are ready to show that for finite clusters
its sets $A$ and $B$ satisfy
\begin{equation}
\label{eq:A>B}
||A|| > ||B||,
\end{equation}
where $||\cdot||$  denotes the number of elements of a set. To this end, let $B + \bm\varepsilon_i$
denote the set of the nodes that results from  translating each $\bm{n}\in B$ by $\bm\varepsilon_i$.
Clearly, $B + \bm\varepsilon_i \subset A$ for all $i$ because otherwise $B$ would contain
a border node. Since $||B + \bm\varepsilon_i || = ||B||$, we arrive at $||A|| \ge ||B||$.
However, $B + \bm\varepsilon_i \neq B + \bm\varepsilon_j$ for any $i\neq j$. This can be justified by noticing
that the largest value of the $i$-th components of these sets differ by 1.
This leads to (\ref{eq:A>B}).

Each node from $A$ must belong to some dimer, and a dimer can cover
only one node from $A$ due to a different parity
of the nodes forming a dimer.
Thus, the other end of a dimer starting in $A$ must belong to $B$.
Then (\ref{eq:A>B}) implies that the dimers must overlap even
though we assumed they do not.
\end{proof}

What will change if we apply the above reasoning to finite systems?
The only significant difference is that now $B + \bm\varepsilon_i$
need not be a subset of $A$, since some of its nodes can be outside $V$
and hence (\ref{eq:A>B}) need not be valid any more.
Actually, since each node in $A$ is connected via a dimer
to a unique node in $B$ and some nodes in $B$ may belong to
other dimers, each coverage of $V$ must satisfy
the relation opposite to (\ref{eq:A>B}),
\begin{equation}
||A|| \le ||B||.
\end{equation}
This, however, is possible only if $B + \bm\epsilon_i \not\subset V$ for some $i$.
In this way we conclude that $B$ touches at least one side of $V$.
We will show that the same is valid for $A$, too.
\begin{lemma} \label{lemma::A-touches-V}
	Set $A$ touches at least one side of $V$.
\end{lemma}
\begin{proof}
Let $\bm{n} \in B$ be a system border node. One of its nearest-neighbors
must be in $A$, otherwise $\bm n$ would lie at a corner of $V$,
which would imply that $A$ is empty though by definition it is not.
Thus, for some $i\neq j$ we have  $\bm{n} + \bm{\epsilon}_j \in A \subset V$
and $\bm{n} + \bm{\epsilon}_i \not\in V$.
Then $\bm{n} + \bm{\epsilon}_j + \bm{\epsilon}_i \not\in V$, that is,
the node $\bm{n} + \bm{\epsilon}_j \in A$ is at the border of $V$.
\end{proof}

The above observations can be generalized as follows.
The condition  $ n_1+\cdots+n_D = m$ used in the proof
of Lemma~\ref{lemma::infinity} defines a ($D-1$)-dimensional hyperplane,
and one can define a total of $2^D$ different hyperplanes by maximizing
the scalar product  ${\bm\sigma}\cdot\bm{n} \equiv \sigma_1 n_1 + \ldots \sigma_D n_D$
over  $\bm{n} \in C$, where $\sigma_i = \pm1$.  Each of them can be equivalently
used in the proof of Lemma~\ref{lemma::infinity} after generalizing the number $m$ used to define sets $A$ and $B$
to
\begin{equation} \label{eq::def-sigma}
  m(C; \bm{\sigma}) = \max_{\bm{n}\in C} \left( \bm{\sigma}\cdot \bm{n} \right).
\end{equation}
They also define $2^D$ half-spaces $\bm{\sigma} \cdot \bm{n} \le m(C; \bm{\sigma})$,
each containing $C$. Their intersection  forms a convex $D$-dimensional
polyhedron $P$.
Its side corresponding to $\bm{\sigma}$ in (\ref{eq::def-sigma})
will be denoted as $\pi_{\bm{\sigma}}$, and the hyperplane it lies on,
$\bm{\sigma} \cdot \bm{n} = m(C; \bm{\sigma})$, by  $\Pi_{\bm{\sigma}}$.
By construction, each  $\pi_{\bm{\sigma}}$ shares at least one node with cluster $C$.
Some of $\pi_{\bm{\sigma}}$ can be degenerated and contain only a single
lattice node---in this case this node must belong to $C$.

The notions of $P$ and its sides $\pi_{\bm{\sigma}}$ are visualized in Fig.~\ref{fig:2D-container}.
Here the system ($V$) is two-dimensional (2D), finite, jammed, and contains three clusters.
The central one (green full symbols) is
tightly encompassed by $2^D=4$  straight lines  $\pm n_1 \pm n_2 = \mbox{const}$.
These lines (or hyperplanes in the general case) are denoted in the paper as $\Pi_{\bm{\sigma}}$.
They form a rectangle (polyhedron) $P$ (green dashed line),
and its sides are denoted as $\pi_{\bm{\sigma}}$ in the paper.
In Fig.~\ref{fig:2D-container} we also distinguished the outer border
of the cluster, $\Gamma'$, which is here made of two disconnected parts
(the connectivity for $\Gamma'$ in 2D is defined via the next-nearest neighbor relation).
The parity of the nodes within each part is the same, but the
parities of disconnected parts are unrelated to each other.
We also enlarged the symbols for the nodes from $C$ that
lie on the edges of both $P$ and $V$. Below we show that each side of $P$
contains at least $D-1$ such ``edge'' points from $V$.
Finally, each dimer in Fig.~\ref{fig:2D-container} belongs to a cluster spanning
opposite sides of the system, and our Theorem~\ref{theorem::1} generalizes this observation
to arbitrary hypercubic systems at jamming.

\begin{figure}
	\includegraphics[width=0.9\columnwidth]{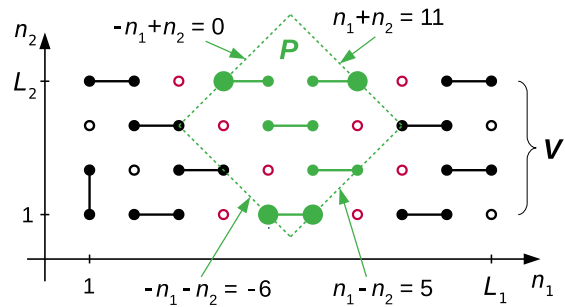}
	\caption{\label{fig:2D-container}
		A jammed coverage of a finite system $V$ ($D=2$, $L_1 = 10$, $L_2 = 4$).
		One cluster is marked with green, and the nodes from the corresponding
		set $\Gamma'$ are in magenta.
		The dotted lines satisfy
		$\pm n_1 \pm n_2  = \mathrm{const.}$ and form a rectangle $P$
		tightly encompassing the cluster. The symbols for the nodes belonging
		to the edges of both $P$ and $V$ are enlarged.
	}
\end{figure}

We thus have $2^D$ sides  $\pi_{\bm\sigma}$ of $P$,
each touching cluster $C$,
and for each of them we can define nonempty sets
$A(\bm{\sigma}) = \{\bm{n} \in C: \bm{\sigma} \cdot \bm{n} = m(C; \bm{\sigma})\}$
and $B(\bm{\sigma}) = \{\bm{n} \in C: \bm{\sigma} \cdot \bm{n} = m(C; \bm{\sigma}) -1 \}$
of cluster's border nodes lying on $\pi_{\bm{\sigma}}$
and their nearest neighbors from $C$, respectively.
These sides are geometrically equivalent---by rotating $P$ we can
transform any of them into any other.  Thus, Lemma~\ref{lemma::A-touches-V}
is applicable to each of them. This leads to

\begin{corollary}
Each side of polyhedron $P$ must by touched by at least one side of hypercuboid $V$.
\end{corollary}

Each side of $V$ is perpendicular to exactly one of the versors $\bm{\epsilon}_i$
and parallel to the others.
Actually, there are always two sides perpendicular to any $\bm{\varepsilon}_i$
and they can be naturally labeled  as $+\bm{\varepsilon}_i$ and $-\bm{\varepsilon}_i$:
the equation of the plane containing the side labeled with $+\bm{\varepsilon}_i$
is $n_i = L_i$, and for  $-\bm{\varepsilon}_i$ it reads $n_i = 1$.

%

\begin{lemma} \label{lemma::shared-vertices}
  Any side $\pi_{\bm{\sigma}}$  of polyhedron $P$ touches
  at least $D-1$ sides of the system $V$.
\end{lemma}

\begin{proof}
Since we are  free to permute the versors $\bm{\epsilon_i}$
and to exchange each $\bm{\epsilon_i}$ with $-\bm{\epsilon_i}$
(mirror reflection), without loss of generality we will prove
Lemma~\ref{lemma::shared-vertices} for  the particular case
$\bm{\sigma} = (+1,\ldots,+1)$. Let $\bm{x}^1_\mathrm{A}$ be a lattice node
that belongs both to the cluster ($C$) and to $\pi_{\bm{\sigma}}$.
This node must belong to the edge of the cluster ($\bm{x}^1_\mathrm{A} \in A$)
and to a dimer. The other node belonging to this dimer must be of the form
$\bm{x}^1_\mathrm{B} \equiv \bm{x}^1_\mathrm{A} - \bm{\epsilon}_i$,
with some $1 \le i \le D$, and belongs to $B$.
For the reasons stated above, we can assume $i=1$  (the dimer
is parallel to $\bm{\epsilon}_1$). There are thus $D-1$
directions orthogonal to the dimer, given by
$\bm{\epsilon}_2,\ldots,\bm{\epsilon}_D$.
Let us distinguish any of them, say, $\bm{\epsilon}_2$.
Consider node
$\bm{x}^2_\mathrm{A} \equiv \bm{x}^1_\mathrm{B} + \bm{\epsilon}_2 =
\bm{x}^1_\mathrm{A}-\bm{\epsilon}_1 + \bm{\epsilon}_2$.
This node belongs to $\Pi_{\bm{\sigma}}$ and is adjacent to a node in $B$,
so either it belongs to $A$ or is outside the system $V$.
In the former case it belongs to a dimer whose
other node is at $\bm{x}_B^2 \equiv \bm{x}_A^2 - \bm{\epsilon}_j$,
with some $j \neq 2$.
We construct
$\bm{x}^3_\mathrm{A} \equiv \bm{x}^2_\mathrm{B} + \bm{\epsilon}_2$
and notice that, just as for  $\bm{x}^2_\mathrm{A}$,
it belongs to $\Pi_{\bm{\sigma}}$, so either it belongs to $A$
or is outside the system $V$.
We can continue this to obtain a sequence of different points
$\{\bm{x}^k_\mathrm{A}\}_{k=1}^n$, $n \ge2$, which are of the form
$\bm{x}^1_\mathrm{A} - \sum_{j=1}^{k-1} \bm{\epsilon}_{i_j} + (k-1)\bm{\epsilon_2}$
with  ${i_j} \in \{1,2,\ldots,D\} \setminus \{2\}$.
Except for the last one, each consecutive element
in this sequence belongs to $A$.
Since the cluster is finite, so is the sequence.
As the last of its elements, denoted as
$\bm{x}^n_\mathrm{A}$, is outside the system and, by construction,
$\bm{x}^n_\mathrm{A} - \bm{\epsilon}_2 = \bm{x}^{n-1}_\mathrm{B}$
belongs to the cluster, these two lattice nodes are on the opposite
sides of the system's side labeled by $\bm{\epsilon}_2$
and so this side touches the cluster.
The same reasoning can be repeated for all other $\bm{\epsilon}_j \neq \bm{\epsilon}_1$
to show that the plane $\pi_{\bm{\sigma}}$ shares a node with at least $D-1$ sides of the system.
The proof is concluded by repeating
the above reasoning for all remaining $\pi_{\bm{\sigma}}$.
\end{proof}

Let us now return for a moment to the case of an infinite
system and formulate an interesting observation:

\begin{lemma} \label{lemma::D-1}
	Any dimer in a jammed configuration of nonoverlapping dimers on an infinite
	$D$-dimedsional hypercubic lattice belongs to a cluster that
	extends to  infinity in at least $D-1$ directions.
\end{lemma}
\begin{proof}
The proof is the same as for Lemma~\ref{lemma::shared-vertices} except
that now the system is infinite and so the sequence $(\bm{x}^k_\mathrm{A})$
must be infinite.
\end{proof}

We are now ready to prove Theorem~\ref{theorem::1}.
Suppose that it is not valid and
that for some jammed configuration of dimers there exists a cluster
that does not touch two opposite sides of a hypercuboid $V$
on a $D$-dimensional hypercubic lattice.
Of the two sides of $V$ orthogonal to each $\bm{\epsilon}_i$
we could select at least one that is not touched by the cluster.
Let us denote it by $\sigma'_i$.
We could then construct the polyhedron $P$ as described earlier and
consider its side $\pi_{\bm{\sigma'}}$
with $\bm{\sigma'} = (\sigma'_1,\ldots,\sigma'_D)$.
By Lemma~\ref{lemma::shared-vertices}, it touches at least $D-1\ge1$
sides of $V$, which contradicts the assumption that it touches none.





\section{Conclusions \label{sec:Conclusions}}

The problem of percolation in the jammed state of the RSA process of $k$-mers
(needles of size $1 \times \dots \times 1 \times k$) is  far easier to formulate
than to treat rigorously.  The conjecture is that each $k$-mer in a jammed
configuration on a hypercubic lattice of space dimension $D \ge 2$ belongs
to a percolating cluster. This conjecture has been already proved
for $D = 2$ and arbitrary  $k$ \cite{Kondrat17}.
Here we have rigorously proven it for $k = 2$ and arbitrary $D$.
We hope  that the methods we have developed will prove useful in proving
the most general case of arbitrary $k$ and $D$.



\section*{References}


%

\end{document}